\documentclass[10pt,journal,compsoc]{IEEEtran}

\ifCLASSOPTIONcompsoc

  \usepackage[nocompress]{cite}
\else
  \usepackage{cite}
\fi

\usepackage{graphicx}
\usepackage[cmex10]{amsmath}
\usepackage{amssymb}
\usepackage{amsthm}
\usepackage{array}
\usepackage{url}
\usepackage{algorithm,algorithmic}
\usepackage{chemarrow}
\usepackage{multirow}
\usepackage{extarrows}
\usepackage{setspace}
\usepackage{booktabs}
\usepackage{tabulary}
\usepackage[shortlabels]{enumitem}
\usepackage{balance}
\usepackage[cmex10]{amsmath}
\usepackage[shortlabels]{enumitem}
\usepackage{amssymb}
\usepackage{amsthm}
\usepackage{multirow}
\usepackage{cite}
\usepackage{color}
\usepackage{setspace}

\makeatletter
\newcommand{\ALOOP}[1]{\ALC@it\algorithmicloop\ #1%
  \begin{ALC@loop}}
\newcommand{\ENDALOOP}{\end{ALC@loop}\ALC@it\algorithmicendloop}

\newtheorem{theorem}{\textbf{\emph{Theorem}}}
\newtheorem{definition}{\textbf{\emph{Definition}}}

\begin{document}

\title{Optimizing Secure Decision Tree Inference Outsourcing}
\author{Yifeng Zheng, Cong Wang, \emph{Fellow, IEEE}, Ruochen Wang, Huayi Duan, and Surya Nepal
\IEEEcompsocitemizethanks{

\IEEEcompsocthanksitem Y. Zheng is with the School of Comptuer Science and Technology, Harbin Institute of Technology, Shenzhen. E-mail: yifeng.zheng@my.cityu.edu.hk.

\IEEEcompsocthanksitem C. Wang is with the Department of Computer Science, City University of Hong Kong, Hong Kong. E-mail: congwang@cityu.edu.hk.

\IEEEcompsocthanksitem R. Wang and H. Duan are with the Department of Computer Science, City University of Hong Kong, Hong Kong. E-mail: ruochwang@gmail.com, hduan2-c@my.cityu.edu.hk.

\IEEEcompsocthanksitem S. Nepal is with Data61, CSIRO, Marsfield NSW 2122, Australia, and also with the Cyber Security Cooperative Research Centre (CRC), Joondalup WA 6027, Australia. Email: surya.nepal@data61.csiro.au.

% \IEEEcompsocthanksitem Corresponding authors: Yifeng Zheng and Cong Wang.
}
}

\IEEEtitleabstractindextext{%
\begin{abstract}
Outsourcing decision tree inference services to the cloud is highly beneficial, yet raises critical privacy concerns on the proprietary decision tree of the model provider and the private input data of the client. In this paper, we design, implement, and evaluate a new system that allows highly efficient outsourcing of decision tree inference. Our system significantly improves upon the state-of-the-art in the overall online end-to-end secure inference service latency at the cloud as well as the local-side performance of the model provider. We first presents a new scheme which securely shifts most of the processing of the model provider to the cloud, resulting in a substantial reduction on the model provider's performance complexities. We further devise a scheme which substantially optimizes the performance for encrypted decision tree inference at the cloud, particularly the communication round complexities. The synergy of these techniques allows our new system to achieve up to $8 \times$ better overall online end-to-end secure inference latency at the cloud side over realistic WAN environment, as well as bring the model provider up to $19 \times$ savings in communication and $18 \times$ savings in computation.
\end{abstract}

}

% make the title area
\maketitle

\IEEEdisplaynontitleabstractindextext

\IEEEpeerreviewmaketitle

\section{Introduction}

Machine learning inference services greatly benefit various kinds of application domains (e.g., healthcare \cite{AzarE13a,libbrecht2015machine,WangZLZL20}, finance \cite{YapOH11,DelenKU13}, and intrusion detection \cite{SindhuGK12,AmorBE04}), and its rapid development has been largely facilitated by cloud computing \cite{AzureML,googleML,awssagemaker} in recent years.
% have been increasingly deployed in many application domains, such as healthcare, finance, and more.
%
% Such a trend has been largely facilitated by the prevalent adoption of cloud computing in recent years.
%
In this emerging machine learning based service paradigm, a model provider can deploy a trained model in the cloud, which can then provide inference services to the clients.
Outsourcing such services to the cloud promises well-understood benefits for both the model provider (\emph{provider} for short) and client, such as scalability, ubiquitous access, and economical cost.

Among others, decisions trees are one of the most popular machine learning models due to its ease of use and effectiveness, and have been shown to benefit real applications like medical diagnosis \cite{AzarE13a,LiangQL19} and credit-risk assessment \cite{YapOH11}.
%
% It is thus the model within our focus in this paper.
%
Briefly, a decision tree is comprised of some internal nodes, which are called decision nodes, and some leaf nodes.
Each decision node is used to compare a threshold with a certain feature in the feature vector, which is the input to decision tree evaluation, and decide the branch to be taken next.
And each leaf node carries a prediction value indicating the inference result.
Decision tree inference over an input feature vector is equivalent to tree traversal starting at the root node and terminating when a leaf node is reached.

While outsourcing the decision tree inference service to the cloud is quite beneficial, it also raises critical privacy concerns on the decision tree model and the input data.
On the provider side, it is widely known that training a high quality model requires a significant amount of investment on datasets (possibly sensitive), resources, and specialized skills. It is thus important that the decision tree is not exposed in the service so that the intellectual property as well as the profitability and competitive advantage of the provider could be respected.
On the client side, the input feature vector may contain sensitive information, e.g., data in medical applications or financial applications.
Overcoming the privacy hurdles is thus of paramount importance to help the provider and client gain confidence in outsourced decision tree inference services.
Towards this challenge, a recent research endeavor has been presented by Zheng et al. \cite{ZhengDWWN20}, which represents the state-of-the-art.
Their design is based on the lightweight additive secret sharing technique and works under a compatible architecture where two cloud servers from independent cloud providers are employed to jointly conduct the decision tree inference in the ciphertext domain. 
As an initial endeavor, however, their design is not fully satisfactory and yet to be optimized in performance, as we detail below.

Firstly, the performance complexity of the provider is dependent on the size of the decision tree as well as the feature vector. 
Specifically, the provider needs to construct and encrypt a binary matrix of size scaling to the product of the number $J$ of decision nodes and the dimension $I$ of the feature vector, so as to support secure feature selection (more details in Section \ref{subsec:secure-input-preparation}). 
Such multiplicative complexity $O(J\cdot I)$ leads to practically unfavorable overhead, which would be further aggravated when the provider needs to outsource multiple decision trees, either for different application domains, or for random forest (an ensemble of decision trees) evaluation.

Secondly, at the cloud side, the phase of secure decision node evaluation has communication rounds linear to the number of bits for value representation.
This is unfavorable in the real-world scenario when the two cloud servers are situated in different geographic regions and communicate over WAN, which is a more reasonable setting than local networks given that the two cloud servers are assumed from different trust domains \cite{ChenPopa20}.

In light of the above observations, In this paper, we present a new highly efficient design for secure decision tree inference outsourcing which significantly improves upon the state-of-the-art.
%
% In this paper, we present a new design that achieves high efficiency in outsourced decision tree inference with privacy.
% %
Our design follows the same architecture of \cite{ZhengDWWN20}, and also makes use of additive secret sharing, yet with significant optimizations to achieve largely boosted performance compared to the state-of-the-art work.

Firstly, we design a new scheme which makes the provider's performance complexity \emph{independent} of the feature vector and thus free of the above multiplicative complexity, through a new re-formulation of the secure feature selection problem.
We make an observation that secure feature selection can indeed be treated as an oblivious array-entry read problem, where the encrypted feature vector could be treated as an encrypted array, and the encrypted index value is used to obliviously select an entry from the array. 
During the procedure, it is required that no information about the feature vector, index value, and selected feature be revealed. 
With this observation, we propose a new secure feature selection design where the provider only needs to construct and encrypt an indexing vector with size $O(J)$, rather than a matrix of size $O(J\cdot I)$ as proposed in \cite{ZhengDWWN20}.

Secondly, we note that the linear communication round complexity of the prior work \cite{ZhengDWWN20} in the secure decision node evaluation phase is due to the secure realization of a ripple carry adder for secret-shared comparison, which faces a delay problem due to sequential procedure of carry computation.
Our observation from the field of digital circuit design is that the carry delay problem presented in the ripple carry adder can be solved via the advanced carry look-ahead adder \cite{DHarris03}.
With this observation, we craft a new design for secure decision node evaluation, through digging deep into the logic and computation of the carry look-ahead adder and appropriately organizing the computation in a secure and efficient manner. 
Our new design achieves a \emph{logarithmic} communication complexity for secure decision node evaluation, gaining superior suitability for practical deployment in WAN environments.
As a concrete example, we are able to significantly reduce the rounds of secure decision node evaluation at the cloud servers from $125$ to $7$ (with the bit length for value representation being $64$), greatly reducing the network latency due to interaction rounds.
We also provide concrete complexity analysis, showing that such significant gain does not sacrifice computational efficiency in terms of the number of secret-shared multiplications.

The synergy of the above optimization techniques lead to a new highly efficient cryptographic inference protocol which achieves a significant reduction on the overall online inference latency at the cloud, as well as a significant boost in the provider's performance, as compared to the state-of-the-art.
We provide formal security analysis of our design under the standard simulation based paradigm.
We implement our system and make deployment on the Amazon cloud for performance evaluation over various decision trees with realistic sizes.
Compared with the state-of-the-art prior work \cite{ZhengDWWN20}, the overall online end-to-end inference latency at the cloud servers over realistic WAN environment is up to $8 \times$ better.
In the meantime, our system offers the provider up to $19 \times$ savings in communication and $18\times$ savings in computation.

\begin{figure}[t!]
\centerline{\includegraphics[width=0.46\textwidth]{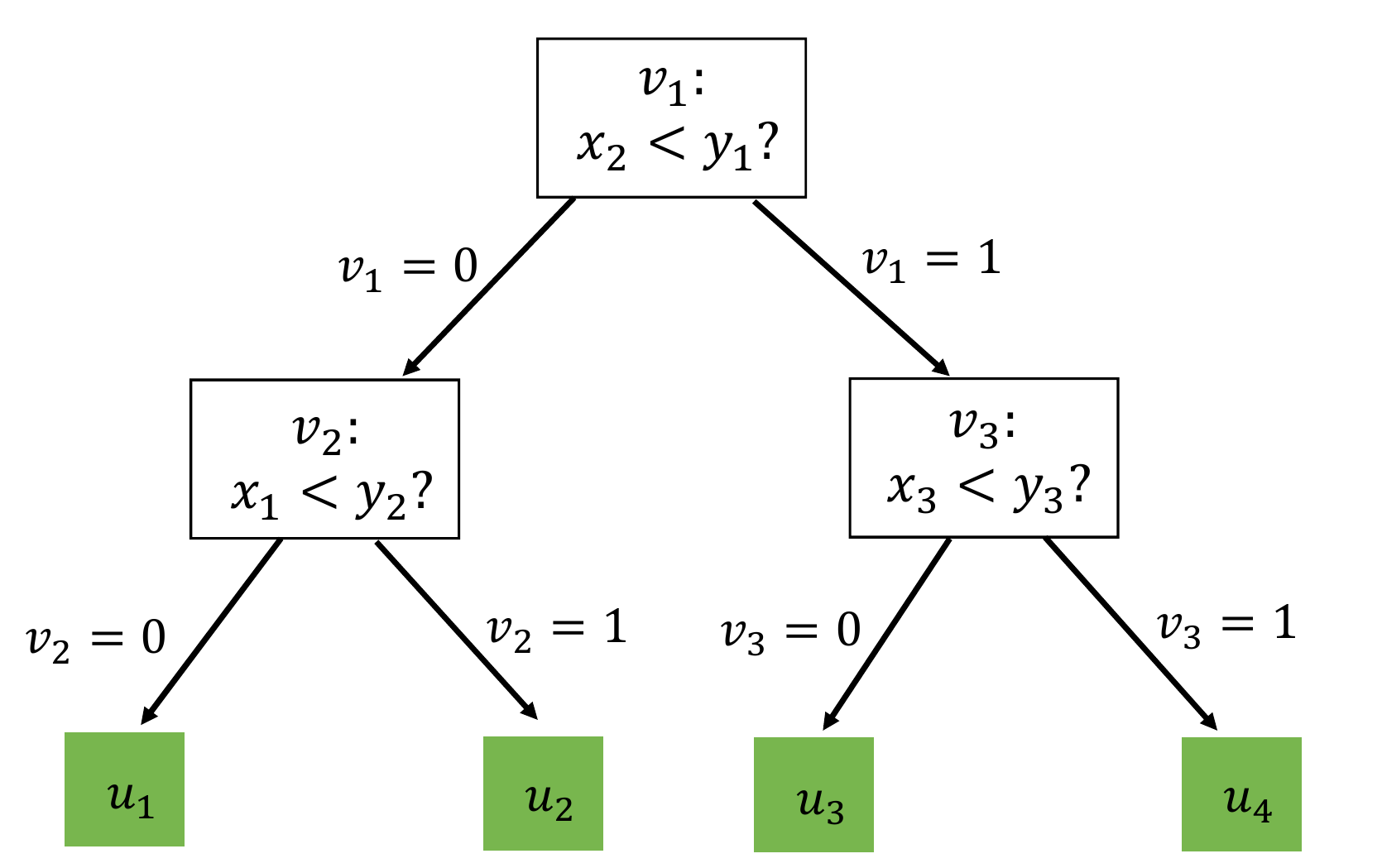}}
\caption{Decision tree illustration.}
\label{fig:decision_tree}
\end{figure}

The rest of this paper is organized as follows.
Section \ref{sec:preliminaries} introduces some preliminaries.
Section \ref{sec:problem-statement} describes the system model and threat model.
Section \ref{sec:secure-design} gives the details of our design.
Section \ref{sec:security-analysis} provides the security analysis.
Section \ref{sec:experiments} shows the experiments.
Section \ref{sec:related_work} discusses the related work.
Section \ref{sec:conclusion} concludes the whole paper.

\section{Preliminaries}
\label{sec:preliminaries}

\subsection{Decision Tree Inference}

Fig. \ref{fig:decision_tree} illustrates a decision tree.
As shown, each internal node (called decision node $\mathcal{D}_j$) is associated with a threshold $y_j$, while each leaf node $\mathcal{L}_z$ is associated with a prediction value $u_z$ indicating the possible inference result.
Hence, given a decision tree with $J$ decision nodes and $Z$ leaf nodes, a threshold vector $\mathbf{y}=\{y_0,\cdots, y_{J-1}\}$ and a prediction value vector $\mathbf{u}=\{u_{0},\cdots, u_{Z-1}\}$ are derived.
The input for decision tree inference is an $I$-dimensional feature vector, denoted by $\textbf{x}=\{x_{0},\cdots, x_{I-1}\}$.
There is an associated input selection mapping $\sigma: j\in \{0, 1,\cdots, J-1\}  \rightarrow i \in \{0,1,\cdots, I-1\}$.
Decision tree inference with $\mathbf{x}$ as input works as follows. 
Firstly, the mapping $\sigma$ is used to select a feature $x_i$ from $\mathbf{x}$ for each $\mathcal{D}_j$. 
Secondly, starting from the root node, the Boolean function $f(x_{\sigma(j)})=(x_{\sigma(j)} < y_j)$ is evaluated at each $\mathcal{D}_j$.
The evaluation result $v_j$ decides whether to next take the left ($v_j=0$) or right ($v_j=1$) branch.
Such evaluation terminates when a leaf node is reached.
The depth $d$ is the length of the longest path between the root node and a leaf node.
Table \ref{table:notations} provides a summary of the key notations. 
Without loss of generality and as the tree structure should be hidden, we will consider complete binary decision trees in our security design, which is also consistent with previous works \cite{DWuFNL16,TaiMZC17,TuenoKK19,Cock17,ZhengDWWN20}.
It is noted that dummy nodes can be simply added to make non-complete decision trees complete \cite{DWuFNL16}.

% Please add the following required packages to your document preamble:

\begin{table}[t!]

\centering
\caption{Key Notations}
\begin{tabular}{@{}c|l@{}}
\toprule
Notation & Description                           \\ \hline

$\mathbf{x}$ & Feature vector\\
$\mathbf{y}$ & Threshold vector \\
$x_i$        & The $i$-th feature in the feature vector                        \\ 
$y_j$        & The threshold at decision node $\mathcal{D}_j$                       \\ 
$d$ & Depth of a decision tree \\
$J $       & Number of decision nodes              \\ 
$I$       & Dimension of feature vector           \\ 
$Z$       & Number of leaf nodes              \\ 
$l$ & Number of bits for value representation \\ 
$v_j$     & Evaluation result at decision node $\mathcal{D}_j$ \\
$u_z$     & Prediction value of leaf node $\mathcal{L}_z$ \\ \bottomrule
\end{tabular}
\label{table:notations}
\end{table}

\subsection{Additive Secret Sharing} 
% The lightweight technique \emph{additive secret sharing} will be used in our design for data encryption on both the provider and the client.
%
Given a value $\alpha \in\mathbb{Z}_{2^l}$, its $2$-of-$2$ additive secret sharing is a pair $([\alpha]_0=\alpha-r$, $[\alpha]_1=r)$, where $r$ is a random value in $\mathbb{Z}_{2^l}$ and the subtraction is done in $\mathbb{Z}_{2^l}$ (i.e., result is modulo $2^l$).
Given either $[\alpha]_0$ or $[\alpha]_1$, the value $\alpha$ is perfectly hidden.
Suppose that two values $\alpha$ and $\beta$ are secret-shared among two parties $\mathcal{P}_0$ and $\mathcal{P}_1$, i.e., $\mathcal{P}_0$ holds $[\alpha]_0$ and $[\beta]_0$ while $\mathcal{P}_1$ holds $[\alpha]_1$ and $[\beta]_1$.
%
% Computation over the secret shares can be conducted as follows.
%
% Firstly,
The secret sharing $[\alpha+\beta]$ (resp. $[\alpha-\beta]$) of $\alpha+\beta$ (resp. $\alpha-\beta$) can be computed locally where each party $\mathcal{P}_i$ ($i\in \{0,1\}$) directly computes $[\alpha+\beta]_{i}=[\alpha]_{i}+[\beta]_{i}$ (resp. $[\alpha-\beta]_{i}=[\alpha]_{i}-[\beta]_{i}$).
Multiplication by a constant $\gamma$ on the value $\alpha$ can also be done locally, i.e., $[\alpha\cdot \gamma]_{i}=\gamma \cdot [\alpha]_{i}$. 
Multiplication over two secret sharings $[\alpha]$ and $[\beta]$ can be supported by using the Beaver's multiplication triple \cite{Beaver91a,Corrigan-GibbsB17}.
That is, given the secret sharing of a multiplication triple $(t_1,t_2,t_3)$ where $t_3=t_1\cdot t_2$, $[\alpha\cdot \beta]$ can be obtained with one round of interaction between the two parties.
In particular, each party $\mathcal{P}_i$ first computes $[e]_i=[\alpha]_i - [t_1]_i$ and $[f]_i=[\beta]_i - [t_2]_i$.
Then, $\mathcal{P}_i$ broadcasts $[e]_i$ and $[f]_i$, and recovers $e$ and $f$.
Given this, each party $P_i$ computes $[\alpha \cdot \beta]_i=i \cdot e \times f+ [t_1]_i \times f + [t_2]_i \times e + [t_3]_i$. 

\section{Problem Statement}
\label{sec:problem-statement}

\subsection{System Model}

Fig. \ref{fig:system_architecture} shows our system architecture, comprised of the provider, the client, and two cloud servers hosted by independent and geographically separated cloud services.
Such architecture follows the state-of-the-art prior work \cite{ZhengDWWN20}.
The provider (e.g., a medical institution) owns a decision tree model and provides inference services to the client with the power of cloud computing, i.e., outsourcing the inference service to the cloud.
Due to concerns on the proprietary decision tree, the provider would only provide an encrypted version.
The client (e.g., a patient) holds a feature vector which may encode private information such as weight, height, heart rate, and blood pressure, and wants to use the intelligent inference service to obtain a prediction about, e.g., her health.
As the feature vector is privacy-sensitive, the client is only willing to provide a ciphertext.

\begin{figure}[t!]
\centerline{\includegraphics[width=0.48\textwidth]{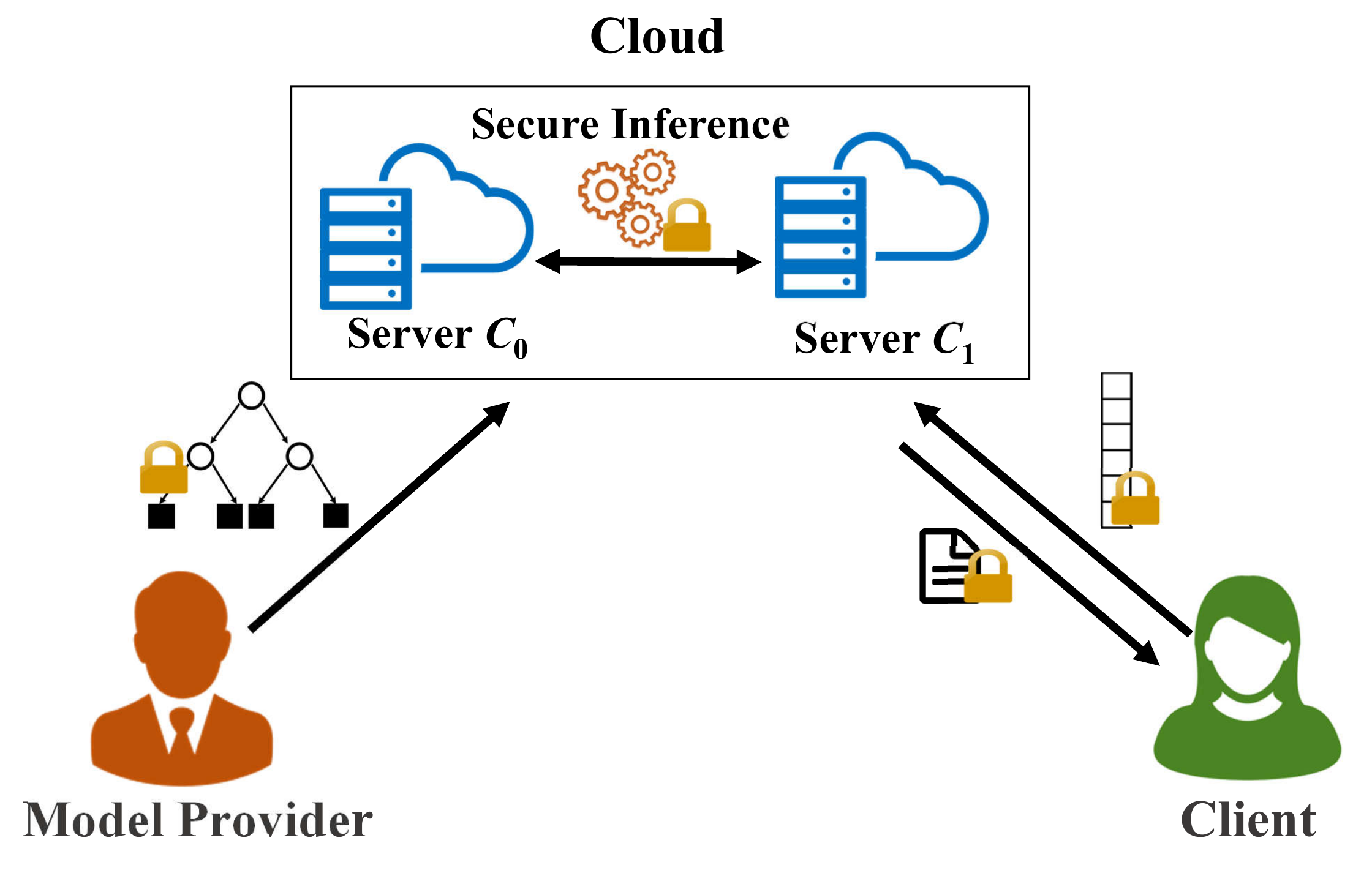}}
\caption{The system architecture.}
\label{fig:system_architecture}
\end{figure}

The power of the cloud is considered to be supplied by the two cloud servers $\mathcal{C}_0$ and $\mathcal{C}_1$, which jointly provide the secure decision tree inference service. 
Such a two-server model not only appeared in the prior work \cite{ZhengDWWN20} on secure outsourced decision tree inference, but has also been recently used to facilitate security designs in different applications \cite{WangWHZR16,MohasselZ17,0002SKG19,ChenPopa20,ZhengDW18}, with tailored use according to problem specifics.
The prominent advantage of such a two-server model is that it allows the provider and the client to go offline after supplying the encrypted inputs, and the secure inference computation is can be fully run at the cloud.
Besides, it is compatible with the working paradigm of additive secret sharing, which is applied for the encryption of the decision tree and feature vector. 
Each cloud server receives shares of the decision tree and feature vector.
They jointly do the processing and produce secret-shared inference result which can be retrieved by the client on demand to reconstruct the inference result.

It is noted that as the two cloud servers are assumed from different trust domains, a practical consideration on realistic deployment is that the two cloud servers be situated in different geographic regions and communicate over a WAN, which is a more reasonable setting compared to local networks.
In this case, the latency due to interactions between the cloud servers should be taken into account as an important factor in the secure system design.

\begin{figure*}[!t]
\centering

\fbox{
  \begin{minipage} [t]{0.8\textwidth}

\textbf{Input:} Secret sharings $[\mathcal{I}_j]$ and $[\mathbf{x}]$.

\textbf{Output:} Secret sharing $[x_{\mathcal{I}_j}]$.

\begin{enumerate}[1:]

\item  Each $\mathcal{C}_m$ creates an array $\mathbf{p}'_m$ where the $i$-th element is $\mathbf{p}'_m[i]=\mathbf{p}_m[i^*_m]+r_m$.
Here, $s_m \leftarrow \mathbb{Z}_{2^l}$ and $r_m \leftarrow \mathbb{Z}_{2^l}$ are random values chosen by $\mathcal{C}_m$; and $i^*_m=((i+s_m)\bmod 2^l) \bmod I$.

\item $\mathcal{C}_0$ chooses a random value $r\leftarrow \mathbb{Z}_{2^l}$ and sends $[\mathcal{I}_j]_0'=[\mathcal{I}_j]_0+r$ to $\mathcal{C}_1$. 

\item $\mathcal{C}_1$ computes $[\mathcal{I}_j]_0'+[\mathcal{I}_j]_1+s_1=\mathcal{I}_j+r+s_1$ and sends it to $\mathcal{C}_0$. 

\item $\mathcal{C}_0$ removes $r$ and produces $i_1'=((\mathcal{I}_j+s_1)\bmod 2^l) \bmod I$.

\item $\mathcal{C}_0$, with $i_1'$ as input, acts as the receiver to run an OT protocol with $\mathcal{C}_1$ to obtain $\mathbf{p}'_1[i_1']$.

\item $\mathcal{C}_1$, in a symmetric manner following Steps 2-5, obtains $\mathbf{p}'_0[i_0']$, where $i_0'=((\mathcal{I}_j+s_0) \bmod 2^l) \bmod I$.

\item $\mathcal{C}_1$ chooses a random value $r'\leftarrow \mathbb{Z}_{2^l}$ and sends $\mathbf{p}^*_0[i'_0]=\mathbf{p}'_0[i'_0]-r_1-r'$ to $\mathcal{C}_0$. Also, $\mathcal{C}_1$ sets $r'$ as its share $[x_{\mathcal{I}_j}]_1$ for the expected feature $x_{\mathcal{I}_j}$.

\item $\mathcal{C}_0$ computes $\mathbf{p}^*_0[i'_0]+ \mathbf{p}'_1[i'_1]-r_0=x_{\mathcal{I}_j}-r'$ and sets the result as its share $[x_{\mathcal{I}_j}]_0$ for $x_{\mathcal{I}_j}$.

\end{enumerate}

\end{minipage}
}
\caption{Secure feature selection for a decision node $\mathcal{D}_j$.}
\label{fig:secure-input-selection}
\end{figure*}

\subsection{Threat Model}

Following prior work on secure outsourced decision tree inference as well as most of existing works on privacy-preserving machine learning \cite{ZhengDWWN20,0002SKG19,MohasselZ17}, we consider a semi-honest adversary setting in our system.
A semi-honest adversary would honestly follow our protocol, yet attempts to infer private information beyond its access rights.
In our system, it is considered that each entity (cloud server, client, provider) might be corrupted by such adversary.
For the cloud server entity, we follow previous works under the two-server model (\cite{MohasselZ17,ZhengDWWN20,ChenPopa20,0002SKG19,ZhengDW18,WangWHZR16}) and assume they are non-colluding.
Namely, the two cloud servers are not corrupted by an adversary at the same time.

Consistent with \cite{ZhengDWWN20}, we consider that the values in the client's feature vector $\mathbf{x}$ as well as the inference result (i.e., the prediction value $u^*$ corresponding to $\mathbf{x}$) should be kept private for the client.
For the provider, there is a need to keep private the proprietary parameters/information of the decision tree model, including each decision node's threshold $y$, the mapping $\sigma$ for feature selection, and the prediction value of each leaf node (except the inference result revealed to the client per inference).
It is also required that the client learns no additional private information about the decision tree other than the prediction value corresponding to her feature vector.
Following prior work \cite{DWuFNL16,ZhengDWWN20}, we assume some generic meta-parameters as public, including the depth $d$, the dimension $I$, and the number $l$ of bits for value representation.
% needed to represent each element in the feature vector and the threshold vector.
We deem dealing with adversarial machine learning attacks out of the scope.

\section{Our Proposed Design}
\label{sec:secure-design}

\subsection{Overview}

Our system is aimed at secure outsourcing of decision tree inference with high efficiency.
Treating local efficiency as the first priority in our design philosophy, we first aim to shift as much processing as possible to the cloud, reducing the local performance complexities (particularly with respect to the provider in our system).
On top of such consideration, we further aim to achieve high efficiency at the cloud through optimizing the processing.
Our deign mainly relies on the delicate use of the lightweight additive secret sharing technique, rather than uses resource-intensive garbled circuits and homomorphic encryption.  

At a high level, our design is comprised of four phases: secure input preparation, secure feature selection, secure decision node evaluation, and secure inference generation.
The secure input preparation phase requires the provider (resp. the client) to encrypt the decision tree (resp. feature vector), and send the ciphertexts to the cloud servers.
The secure feature selection phase is to securely select for each decision node a certain feature from the feature vector, in such a way that the cloud servers are oblivious to the mapping between decision nodes and features.
The secure decision node evaluation phase securely evaluates the Boolean function at each decision node and output the ciphertext of the evaluation result.
The secure inference generation phase is to leverage the ciphertexts of the evaluation results at decision nodes to generate the ciphertext of the ultimate decision tree inference result, which can then be retrieved by the client for recovery.
Note that following the previous work \cite{ZhengDWWN20}, we assume that the data-independent multiplication triples are pre-generated and made available to the two cloud servers for use in our design, which can be efficiently achieved via a semi-honest third party \cite{RiaziWTS0K18,ZhengDWWN20}. 
Our focus is on the latency-sensitive online inference procedure. 

\subsection{Secure Input Preparation}
\label{subsec:secure-input-preparation}

The client encrypts her feature vector $\mathbf{x}$ via additive secret sharing applied in an element-wise manner. 
In particular, the client generates two secret shares: $[\mathbf{x}]_0=\mathbf{x}-\mathbf{r}$ and $[\mathbf{x}]_1=\mathbf{r}$.
For the provider, he encrypts the decision tree as follows.
Firstly, the vector $\mathbf{y}$ of thresholds at decision nodes and the vector $\mathbf{u}$ of prediction values at leaf nodes are encrypted through additive secret sharing, with the secret shares $[\mathbf{y}]_0$, $[\mathbf{y}]_1$, $[\mathbf{u}]_0$, and $[\mathbf{u}]_1$ produced.
Then, we need to consider how to properly encrypt the mapping $\sigma$ which is used for feature selection.

We note that the prior work \cite{ZhengDWWN20} constructs a binary matrix of size $J \times I$ in such a manner that the $j$-th row vector is a binary vector with $I$ elements where all are $0$ except for the one at position $\sigma(j)$ being set to $1$.
In this way, feature selection is then realized via matrix-vector multiplication between the binary matrix and the feature vector, which can be securely supported under additive secret sharing.
Unfortunately, such an approach imposes on the provider multiplicative $O(J \cdot I)$ performance complexity which depends on the number $J$ of decision nodes as well as the dimension $I$ of the feature vector.

Differently, in order to minimize the costs of the provider, our new insight is to instead construct an index vector $\mathcal{I}$ which is comprised of the selection index values for decision nodes and thus the complexity only depends on the number $J$ of decision nodes, i.e., $O(J)$. 
The selection index value for a decision node $\mathcal{D}_j$ is represented as $\mathcal{I}_j\in [0,I-1]$.
The secure usage of this index vector will be described shortly in the phase of secure feature selection. 
The provider also encrypts this index vector via additive secret sharing and produces $[\mathcal{I}]_0$ and $[\mathcal{I}]_1$.
After the above processing, the provider sends the shares $[\mathbf{y}]_m$, $[\mathbf{u}]_m$, and $[\mathcal{I}]_m$ to each cloud server $\mathcal{C}_m$ ($m\in \{0,1\}$).

\begin{figure*}[t!]
\centerline{\includegraphics[width=0.5\textwidth]{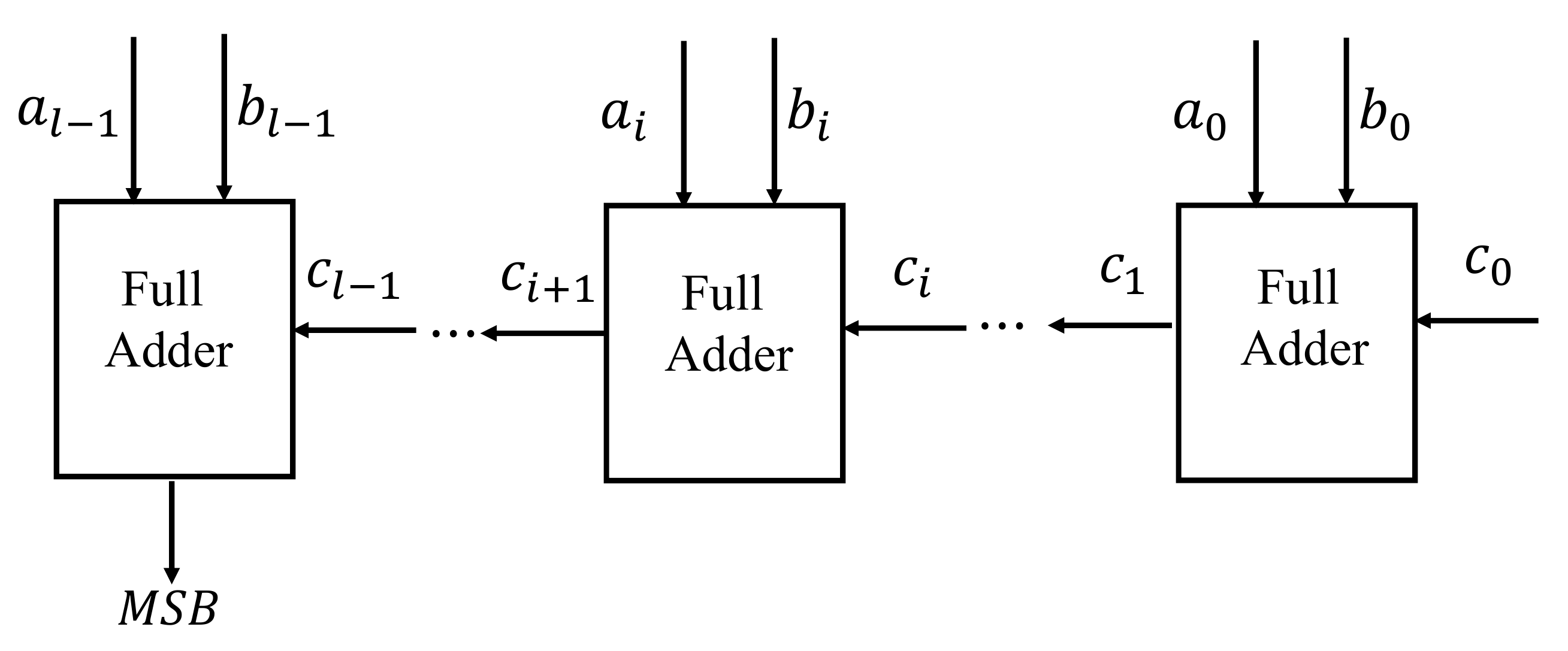}}
\caption{Illustration of a ripple carry adder logic for MSB computation.}
\label{fig:ripple_carry_adder}
\end{figure*}

\begin{figure*}[!t]
\centering
  \begin{minipage}[t]{0.28\linewidth}
    \centering
    \includegraphics[width=\linewidth]{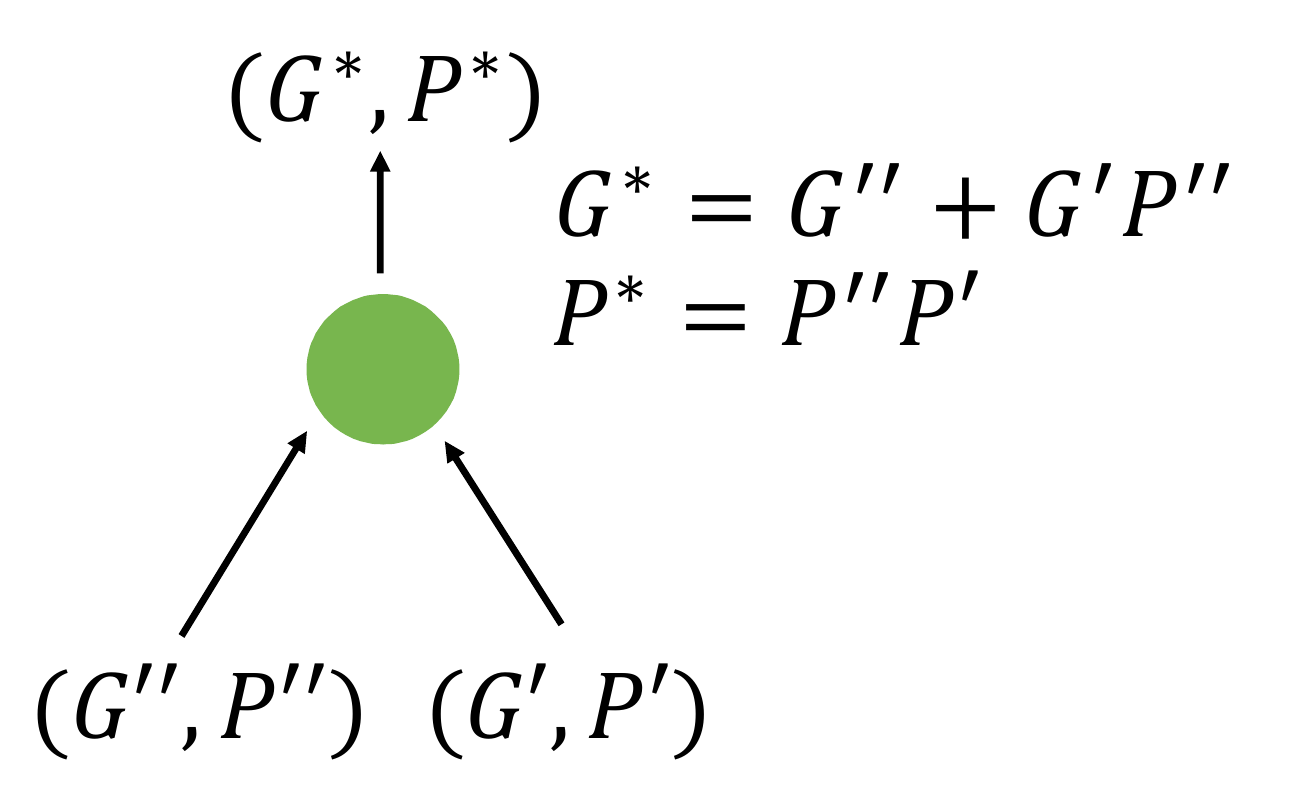}\\\footnotesize{(a)}
  \end{minipage}%
  \begin{minipage}[t]{0.42\linewidth}
    \centering
    \includegraphics[width=\linewidth]{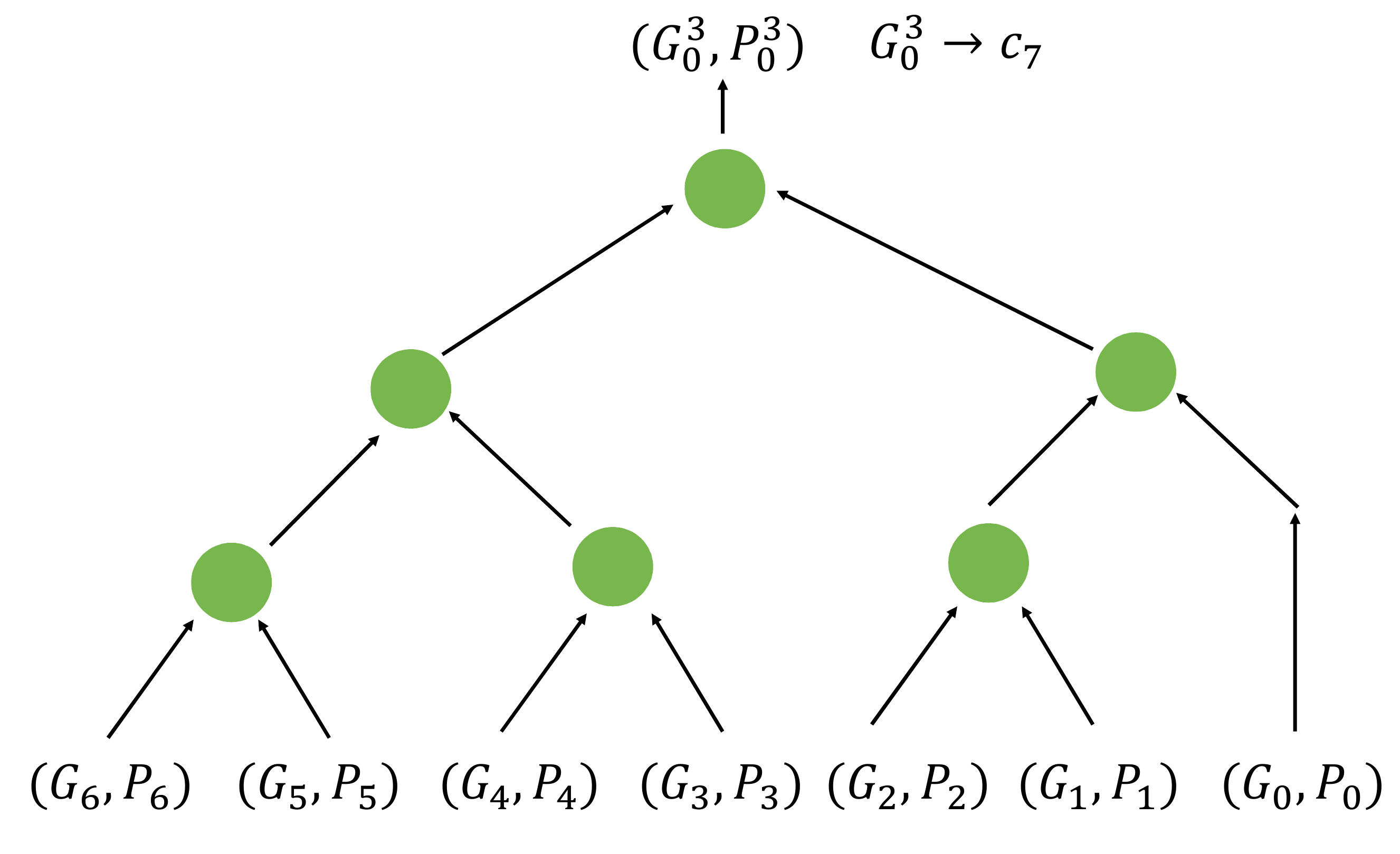}\\\footnotesize{(b)}
  \end{minipage}
  \caption{(a) Illustration of the defined binary operator; (b) Illustration of carry calculation over $8$-bit inputs under the carry look-ahead adder.}
  \label{fig:carry_look_ahead}
\end{figure*}

\subsection{Secure Feature Selection}

Upon receiving the shares of the client's feature vector and the provider's decision tree, the cloud servers first perform secure feature selection which produces the secret sharing of a certain feature for each decision node, based on the encrypted index vector $\mathcal{I}$.
Note that hereafter all arithmetic operations are conducted by default in the ring $\mathbb{Z}_{2^l}$, unless otherwise stated.

We now describe how secure feature selection is achieved in our design.
For each decision node $\mathcal{D}_j$, the processing of secure feature selection would require as input the shares of the client' feature vector $\mathbf{x}$ and the shares of the corresponding index value $\mathcal{I}_j$ in the index vector $\mathcal{I}$.
The output is the secret sharing of the selected feature $x_{\mathcal{I}_j}$ for the decision node $\mathcal{D}_j$.

To accomplish this functionality for our secure outsourced decision tree services, we make an observation that this indeed can be treated as oblivious array-entry read, where the encrypted feature vector could be treated as an encrypted array, and the encrypted index value is used to obliviously select an entry from the array. 
We leverage this observation and identify that an approach from a very recent work \cite{CurranLGPD19} is suited for our scenario.
Using this approach as a basis, we devise the scheme for secure feature selection, which is given in Fig. \ref{fig:secure-input-selection}.

The intuition is as follows.
For ease of notation, we let $\mathbf{p}_0$ (resp. $\mathbf{p}_1$) denote the share of the feature vector $[\mathbf{x}]_0$ (resp. $[\mathbf{x}]_1$) at the cloud server $\mathcal{C}_0$ (resp. $\mathcal{C}_1$).
Each cloud server $\mathcal{C}_m$ ($m\in \{0,1\}$) first creates a new array $\mathbf{p}'_m$ which is derived from $\mathbf{p}_m$ by shifting its indices and entries under fixed random values (per feature selection).
Then, given the secret sharing of a target index value $\mathcal{I}_j$, $\mathcal{C}_0$ engages in an interaction with $\mathcal{C}_1$ so as to receive the entry located at the shifted index $((\mathcal{I}_j+s_1) \bmod 2^l) \bmod I$ in $\mathbf{p}'_1$ which corresponds to $\mathcal{I}_j$.
Since the random value $s_1$ in the shifted index is only known to $\mathcal{C}_1$ and the corresponding entry is masked by a random value $r_1$, $\mathcal{C}_0$ is oblivious to the original index $\mathcal{I}_j$ as well as the share $\mathbf{p}_1[\mathcal{I}_j]$ held by $\mathcal{C}_1$.
Similarly, $\mathcal{C}_1$ obtains the entry located at the shifted index in $\mathbf{p}'_0$ which corresponds to $\mathcal{I}_j$, while learning no information about the plain index $\mathcal{I}_j$ and the share of the corresponding entry value held by $\mathcal{C}_0$.
Next, the two cloud servers engage in an interaction which essentially performs secret re-sharing so as to obtain the secret sharing of the selected feature $x_{\mathcal{I}_j}$.

\subsection{Secure Decision Node Evaluation}

With the secret sharings of the threshold $y_j$ and selected feature $x_{\mathcal{I}_j}$ at each decision node $\mathcal{D}_j$, the cloud servers now perform secure decision node evaluation.
As this basically requires secure comparison of secret-shared values, the prior work \cite{ZhengDWWN20} transforms the problem of secure decision node evaluation to a simplified bit extraction problem in the secret sharing domain.
The key idea is to securely extract the most significant bit (MSB) of the subtraction result $\Delta=y_j-x_{\mathcal{I}_j}$ as the evaluation result at decision node $\mathcal{D}_j$. 

\begin{figure*}[!t]
\centering
\fbox{
  \begin{minipage} [t]{0.85\textwidth}

\textbf{Input:} Secret sharings $[y_j]$ and $[x_{\mathcal{I}_j}]$.

\textbf{Output:} Secret sharing $\left\langle{v_j}\right\rangle$.

\begin{enumerate}[1:]
\item $\mathcal{C}_{m}$ computes $[\Delta]_{m}=[y_j]_{m}-[x_{\mathcal{I}_j}]_{m}$.
\\
// Secure MSB extraction (with $l=64$ assumed; $\left\langle{\cdot}\right\rangle$ denotes sharing over $\mathbb{Z}_2$)

\item Let $a$ (resp. $b$) represent the share $[\Delta]_0$ (resp. $[\Delta]_1$), with the bit string being $a_{l-1},\cdots, a_0$ (resp. $b_{l-1},\cdots, b_0$).
%
% Let $b$ denote $\mathcal{C}_1$'s share $[d]_1$, with the bit string being $b_{l-1},\cdots, b_0$.    
%
Let $\left\langle{a_q}\right\rangle$ be defined as $(\left\langle{a_q}\right\rangle_0=a_q,\left\langle{a_q}\right\rangle_1=0)$ and $\left\langle{b_q}\right\rangle$ as $\{ \left\langle{b_q}\right\rangle_0=0,\left\langle{b_q}\right\rangle_1=b_q \}$, where $q\in [0,l-1]$.
Also, let $\left\langle{w_q} \right\rangle$ be defined as $\{\left\langle{w_q}\right\rangle_0=a_q, \left\langle{w_q}\right\rangle_1=b_q\}$.

// Setup round for secure carry computation (\emph{SCC}):

\item Compute $\left\langle{G_q}\right\rangle=\left\langle{a_q}\right\rangle\cdot \left\langle{b_q}\right\rangle$, for $q\in[0,l-1]$

\item Compute $\left\langle{P_q}\right\rangle=\left\langle{a_q}\right\rangle+\left\langle{b_q}\right\rangle$, for $q\in[0,l-1]$

// \emph{SCC} round 1 (with $l=64$ as example):

\item Compute $(\left\langle{G^1_0}\right\rangle, \left\langle{P^1_0}\right\rangle)=(\left\langle{G_0}\right\rangle, \left\langle{P_0}\right\rangle)$

\item For $k\in \{1,\cdots, 31\}$
\begin{enumerate}

\item Compute $(\left\langle{G^1_k}\right\rangle, \left\langle{P^1_k}\right\rangle)=(\left\langle{G_{2\cdot k}}\right\rangle, \left\langle{P_{2\cdot k}}\right\rangle) \tilde{\diamond} (\left\langle{G_{2\cdot k-1}}\right\rangle, \left\langle{P_{2\cdot k-1}}\right\rangle)$

\end{enumerate}

// \emph{SCC} round 2:

\item For $k\in \{0,\cdots, 15\}$

\begin{enumerate}
\item Compute $(\left\langle{G^2_k}\right\rangle, \left\langle{P^2_k}\right\rangle)=(\left\langle{G^1_{2\cdot k+1}}\right\rangle, \left\langle{P^1_{2\cdot k+1}}\right\rangle) \tilde{\diamond} (\left\langle{G^1_{2\cdot k}}\right\rangle, \left\langle{P^1_{2\cdot k}}\right\rangle)$
\end{enumerate}

// \emph{SCC} round 3: 
\item For $k\in \{0,\cdots, 7\}$

\begin{enumerate}
\item Compute $(\left\langle{G^3_k}\right\rangle, \left\langle{P^3_k}\right\rangle)=(\left\langle{G^2_{2\cdot k+1}}\right\rangle, \left\langle{P^2_{2\cdot k+1}}\right\rangle) \tilde{\diamond} (\left\langle{G^2_{2\cdot k}}\right\rangle, \left\langle{P^2_{2\cdot k}}\right\rangle)$
\end{enumerate}

// \emph{SCC} round 4: 
\item For $k\in \{0,\cdots, 3\}$

\begin{enumerate}
\item Compute $(\left\langle{G^4_k}\right\rangle, \left\langle{P^4_k}\right\rangle)=(\left\langle{G^3_{2\cdot k+1}}\right\rangle, \left\langle{P^3_{2\cdot k+1}}\right\rangle) \tilde{\diamond} (\left\langle{G^3_{2\cdot k}}\right\rangle, \left\langle{P^3_{2\cdot k}}\right\rangle)$
\end{enumerate}

// \emph{SCC} round 5: 
\item For $k\in \{0, 1\}$

\begin{enumerate}
\item Compute $(\left\langle{G^5_k}\right\rangle, \left\langle{P^5_k}\right\rangle)=(\left\langle{G^4_{2\cdot k+1}}\right\rangle, \left\langle{P^4_{2\cdot k+1}}\right\rangle) \tilde{\diamond} (\left\langle{G^4_{2\cdot k}}\right\rangle, \left\langle{P^4_{2\cdot k}}\right\rangle)$
\end{enumerate}

// \emph{SCC} round 6: 
\item Compute $\left\langle{G^6_0}\right\rangle=\left\langle{G^5_1}\right\rangle+\left\langle{G^5_0}\right\rangle \cdot \left\langle{P^5_1}\right\rangle=\left\langle{c_{l-1}}\right\rangle$

\item Compute $\left\langle{v_j}\right\rangle=\left\langle{w_{l-1}}\right\rangle+ \left\langle{c_{l-1}} \right\rangle$.

\end{enumerate}
\end{minipage}
}

\caption{Secure evaluation of a decision node.}
\label{fig:secure-decision-node-evaluation}
\end{figure*}

Despite the effectiveness, their solution is limited in that it poses linear $O(l)$ round complexity.
This would lead to high performance overhead in the realistic scenario where the two cloud servers are situated in different geographic regions and connected over WAN, which is a more reasonable setting than local networks given that the two cloud servers are assumed to be non-colluding \cite{ChenPopa20}.
The basic idea in \cite{ZhengDWWN20} is to implement a $l$-bit full adder logic in the secret sharing domain.
Specifically, the shares of the difference value $\Delta$ at the two cloud servers are represented in bitwise form respectively.
Then, a $l$-bit full adder logic is applied to add in the secret sharing domain the two binary inputs in a bitwise manner, where carry bits are calculated and propagated, and finally produce the MSB of the difference value $\Delta$.
We note that the work \cite{ZhengDWWN20} uses a classical and standard adder logic called ripple carry adder, as shown in Fig. \ref{fig:ripple_carry_adder}.

In the ripple carry adder, for each full adder, the two bits that are to be added are available instantly. 
However, each full adder has to wait for the carry input to arrive from its previous adder.
This means that the carry input for the full adder producing the MSB should wait after the carry has rippled through all previous full adders.
Note that computing the carry output of each full adder in the secret sharing domain requires interactions between the two cloud servers, thus leading to $O(l)$ round complexity.

Our design follows \cite{ZhengDWWN20} in terms of the same strategy of secure MSB bit extraction for secure decision node evaluation, yet aims to reduce the round complexity.
Through the design introduced below, we manage to reduce the round complexity from linear to logarithmic.
Our observation is that the use of the more advanced carry look-ahead adder can solve the carry delay problem presented in the ripple carry adder \cite{DHarris03}.
At a high level, a carry look-ahead adder is able to calculate the carry in advance based on only the input bits.
It works as follows.
Firstly, two terms are defined for the carry look-ahead adder: the carry generate signal $G_i$ and the carry propagate signal $P_i$, where $G_i=a_i\cdot b_i$ and $P_i=a_i+b_i$.
Note that these two terms are only based on the input bits and can be computed instantly given the input bits.
Then, the original carry calculation $c_{i+1}=a_i\cdot b_i+(a_i+b_i)\cdot c_i$ as in the ripple carry adder can then be re-formulated as $c_{i+1}=G_i+P_i\cdot c_i$.
Such re-formulation allows a carry to be computed without waiting for the carry to ripple through all previous stages, as demonstrated by the following example (a $4$-bit carry look-ahead adder):

\begin{enumerate}

\item $c_1=G_0+P_0\cdot c_0=G_0$;
\item $c_2=G_1+P_1\cdot c_1=G_1+P_1\cdot G_0$;
\item $c_3=G_2+P_2\cdot c_2=G_2+P_2\cdot (G_1+P_1\cdot G_0)$; 
\item $c_4=G_3+P_3\cdot c_3=G_3+P_3\cdot (G_2+P_2\cdot (G_1+P_1\cdot G_0))$.
\end{enumerate}

It can be seen that each carry can be computed without waiting for the calculation of all previous carries.

\begin{figure*}[!t]
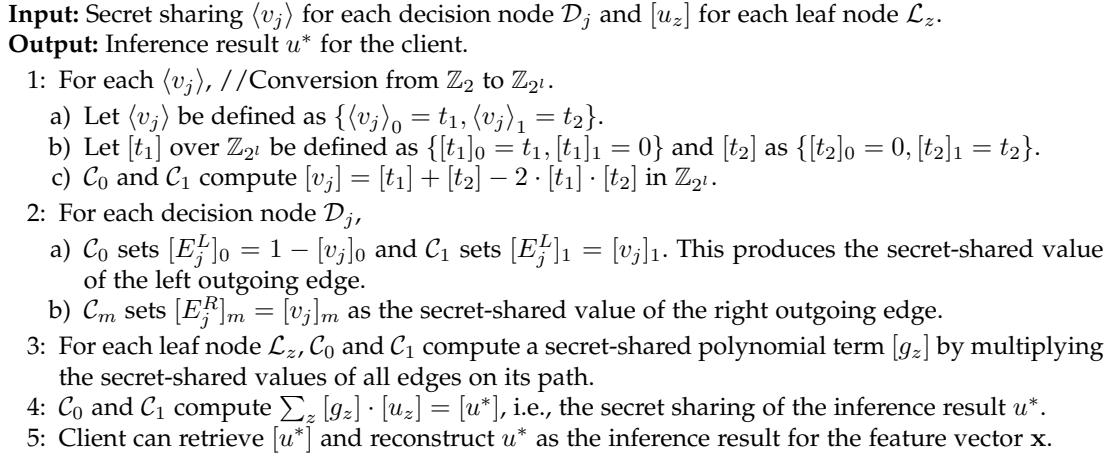

\centering

\fbox{
  \begin{minipage} [t]{0.8\textwidth}

\textbf{Input:} Secret sharing $\left\langle{v_j}\right\rangle$ for each decision node $\mathcal{D}_j$ and $[u_z]$ for each leaf node $\mathcal{L}_z$.

\textbf{Output:} Inference result $u^*$ for the client.

\begin{enumerate}[1:]

\item For each $\left\langle{v_j}\right\rangle$, //Conversion from $\mathbb{Z}_2$ to $\mathbb{Z}_{2^l}$. 

\begin{enumerate}
	\item Let $\left\langle{v_j}\right\rangle$ be defined as $\{\left\langle{v_j}\right\rangle_0=t_1,\left\langle{v_j}\right\rangle_1=t_2\}$. 

	\item Let $[t_1]$ over $\mathbb{Z}_{2^l}$ be defined as $\{[t_1]_0=t_1,[t_1]_1=0\}$ and $[t_2]$ as $\{[t_2]_0=0,[t_2]_1=t_2\}$.

\item $\mathcal{C}_0$ and $\mathcal{C}_1$ compute $[v_j]=[t_1]+[t_2]-2\cdot [t_1]\cdot [t_2]$ in $\mathbb{Z}_{2^l}$.
\end{enumerate}

\item For each decision node $\mathcal{D}_j$,

\begin{enumerate}

\item $\mathcal{C}_0$ sets $[E^L_{j}]_0=1-[v_j]_0$ and $\mathcal{C}_1$ sets $[E^L_{j}]_1=[v_j]_1$. This produces the secret-shared value of the left outgoing edge.

\item $\mathcal{C}_{m}$ sets $[E^R_{j}]_{m}=[v_j]_{m}$ as the secret-shared value of the right outgoing edge.
\end{enumerate}

\item For each leaf node $\mathcal{L}_z$, $\mathcal{C}_0$ and $\mathcal{C}_1$ compute a secret-shared polynomial term $[g_z]$ by multiplying the secret-shared values of all edges on its path.

% \item $\mathcal{C}_0$ and $\mathcal{C}_1$ compute $\sum\nolimits_z{[g_z]}=[u^*]$, i.e., the secret sharing of the inference result $u^*$.

\item $\mathcal{C}_0$ and $\mathcal{C}_1$ compute $\sum\nolimits_z{[g_z]\cdot [u_z]}=[u^*]$, i.e., the secret sharing of the inference result $u^*$.

\item Client can retrieve $[u^*]$ and reconstruct $u^*$ as the inference result for the feature vector $\mathbf{x}$.

\end{enumerate}
\end{minipage}
}
\caption{Secure inference generation.}
\label{fig:polynomial-secure-classification-gen}
\end{figure*}

With the application of the carry look-ahead adder for MSB computation for secure decision node evaluation, we only need to focus on the calculation of the carry $c_{l-1}$ (e.g., $c_3$ in the above example for $4$-bit inputs).
That is, after computing $c_{l-1}$, the MSB can be derived via $MSB=a_{l-1}+b_{l-1}+c_{l-1}$ .
Note that $c_{l-1}=G_{l-2}+P_{l-2}\cdot G_{l-3}+\cdots +P_{l-2}\cdots P_{1}\cdot G_0$.
We now need to consider how to properly organize the computation of the carry $c_{l-1}$ so that we could effectively achieve $O(\log_2 l)$ round complexity.  
Our observation is that such a computation can be supported by forming a binary tree over the carry generate terms, carry propagate terms, and a binary operator $\diamond$ (illustrated in Fig. \ref{fig:carry_look_ahead}(a)) defined as: $(G^*,P^*)=(G'',P'')\diamond (G',P')$, where $G^*=G''+G'\cdot P''$ and $P^*=P''\cdot P'$.
For demonstration of this idea, we show in Fig. \ref{fig:carry_look_ahead}(b) an example on how the carry bit essential for MSB computation can be computed in a recursive manner based on a tree structure, in the case of $8$-bit inputs.
As shown, a pair of the carry generate and propagate terms is put as a leaf node of the tree.
Then, the processing is done upwards attributing to each internal node the value corresponding to the application of the operator $\diamond$ between its two children.
Such processing leads to $\lceil{\log_2 l}\rceil$ rounds in computing the essential carry $c_{l-1}$.

For simplicity of illustration, we show here the details for the case of $4$-bit inputs to concretely demonstrate the computation.
In the first round, the following terms are computed: $(G^1_1,P^1_1)=(G_2,P_2)\diamond (G_1,P_1)$, $(G^1_0,P^1_0)=(G_0,P_0)$.
In the second round, the following term is computed: $(G^2_0, P^2_0)=(G^1_1,P^1_1)\diamond (G^1_0,P^1_0)$.
Based on the definition of the binary operator $\diamond$, we first have $G^1_1=G_2+G_1\cdot P_2,P^1_1=P_2\cdot P_1$.
Then, we have $G^2_0=G^1_1+G^1_0\cdot P^1_1=G_2+G_1\cdot P_2+G_0\cdot P_2P_1$, which corresponds to the carry $c_3$.

With all the above insights  in mind, we now elaborate on how to support secure decision node evaluation by taking advantage of the carry look-ahead adder in the ciphertext domain.
From the definition of the operator $\diamond$, it is noted that only addition and multiplication are required (in $\mathbb{Z}_2$).
So it is easy to see this operator can be securely realized in the ciphertext domain via secret-shared addition and multiplication.
We denote the secure realization of the operator as $(\left\langle{G^*}\right\rangle,\left\langle{P^*}\right\rangle)=(\left\langle{G''}\right\rangle,\left\langle{P''}\right\rangle) \tilde{\diamond} (\left\langle{G'}\right\rangle,\left\langle{P'}\right\rangle)$, where $\left\langle{\cdot}\right\rangle$ denotes secret sharing in $\mathbb{Z}_2$.
Note that each call of the operator needs $2$ parallel secret-shared multiplications.

The details of secure decision node evaluation are provided in Fig. \ref{fig:secure-decision-node-evaluation}.
It is noted that for simplicity and without loss of generality, we demonstrate the procedure assuming that $l=64$, which is the practical parameter setting to be used in our experiments, and also consistent with the state-the-art work \cite{ZhengDWWN20}.
Also, to clearly show the computation that can be done in parallel in each round of secure carry computation, we intentionally avoid the use of nested loops.
%
% Note that the addition and multiplication in $\mathbb{Z}_2$ will be indeed realized via XOR operation and AND operation respectively in our experiments.
%
% It is noted that after the secret sharing of the MSB is produced in $\mathbb{Z}_2$ (i.e., Step 12 in Fig. \ref{fig:secure-decision-node-evaluation}), conversion to the original secret sharing domain $\mathbb{Z}_2$ is performed so as to support subsequent computation in the secure inference result generation phase.
%
From the procedure shown in Fig. \ref{fig:secure-decision-node-evaluation}, we can see that for the practical setting $l=64$, only $7$ rounds are required through our design for obtaining the secret sharing $\left\langle{v}\right\rangle$ of the comparison result, in comparison with $125$ rounds in the state-of-the-art work \cite{ZhengDWWN20}.

We also point out that the improvement on communication rounds does not sacrifice the computation efficiency in terms of number of multiplications (in $\mathbb{Z}_2$).
Through analysis, our new design requires $3l-5$ ($187$ for $l=64$) multiplications while the prior work requires $3l-5$ ($187$ for $l=64$) multiplications as well.
We remark that although in principle the carry look-ahead adder has higher circuit complexity than the ripple carry adder when computing \emph{all} carries is required, our design only needs the computation of the essential carry $c_{l-1}$.
This accounts for why our new design does not enlarge the computation cost in secure decision node evaluation compared with \cite{ZhengDWWN20}.

\subsection{Secure Inference Generation}

With the secret-shared evaluation results available at each decision node, we now describe how to leverage them to enable the two cloud servers to generate the encrypted inference result.
We note that there are two approaches on how to use the decision node evaluation results \cite{ZhengDWWN20}: a path cost-based approach and a polynomial-based approach.
From the perspective of client cost, the main difference between these approaches is that the path cost-based approach imposes on the client high communication complexity exponentially scaling with the tree depth, while the polynomial-based approach only incurs constant $O(1)$ and minimal communication cost (just two shares) for the client. 

Given that high local efficiency is our first priority, our system makes use of the polynomial-based approach.
This approach works as follows.
Starting from the root node, the left outgoing edge of each decision node $\mathcal{D}_j$ is assigned the value $1-v_j$ (denoted as $E^L_{j}=1-v_j$), while the right outgoing edge is assigned the value $v_j$ (denoted as $E^R_{j}=v_j$).
Then, a term $g_z$ is computed for each path by multiplying the edge values of that path.
As mentioned before, only the term for the path leading to the leaf node carrying the inference result will have the value $1$, and all other terms are $0$.
Then, we can proceed by multiplying each term $g_z$ with the prediction value $u_z$ of each corresponding leaf node and computing the sum, i.e., $u^*  = \sum\nolimits_z {u_z g_z }$, which will lead to the expected inference result $u^*$.

We use the decision tree in Fig. \ref{fig:decision_tree} as an example to concretely demonstrate the computation.
Firstly, there are four terms $\{g_z\}^4_{z=1}$ given that the depth is $2$ and thus four paths/leaf nodes.
These terms are computed as follows: $g_1=(1-v_1)\cdot (1-v_2)$, $g_2=(1-v_1)\cdot v_2$, $g_3=v_1\cdot (1-v_3)$, and $g_4=v_1\cdot v_3$.
Suppose that the feature vector is evaluated along the path of the leaf node $\mathcal{L}_3$. 
We have $v_1=1$ and $v_3=0$, and so we have $g_1=0 \cdot (1-v_2)=0$, $g_2=0 \cdot v_2=0$, $g_3= 1\cdot (1-0)=1$, and $g_4=1\cdot 0=0$.
It can be seen that all the terms except the term $g_3$ corresponding to $\mathcal{L}_3$ has zero value.
So we have $\sum\nolimits^4_{z=1} {u_z g_z }=u_3$, obtaining the expected inference result.
The secure realization of this approach in our system for secure inference generation basically follows that of \cite{ZhengDWWN20}.
For completeness, we give the details of secure inference generation in Fig. \ref{fig:polynomial-secure-classification-gen}, which realizes polynomial mechanism introduced above in the secret sharing domain. 

\section{Security Analysis}
\label{sec:security-analysis}

We define and prove the security of our protocol following the standard simulation-based paradigm. 
We start with defining the ideal functionality which captures the desired security properties for outsourced decision tree inference, with regard to the threat model mentioned above.
We then give the formal security definition under the ideal functionality and show that our protocol securely realizes the ideal functionality.
In what follows, we define the ideal functionality for the secure outsourced decision tree inference service targeted in this paper.

\begin{definition}
The ideal functionality $\mathcal{F}_{\mathsf{SecODT}}$ of the outsourced decision tree inference service is formulated as follows.

\begin{enumerate}[-]

\item \textbf{Input.} The input to the $\mathcal{F}_{\mathsf{SecODT}}$ consists of the decision tree $\mathcal{T}$ from the provider and the feature vector $\mathbf{x}$ from the client. The two cloud servers $\mathcal{C}_0$ and $\mathcal{C}_1$ provide no input to the $\mathcal{F}_{\mathsf{SecODT}}$.

\item \textbf{Computation.} Upon receiving the above input, the $\mathcal{F}_{\mathsf{SecODT}}$ performs decision tree inference and produces the inference result denoted as $\mathcal{T}(\mathbf{x})$.

\item \textbf{Output.} The $\mathcal{F}_{\mathsf{SecODT}}$ outputs the inference result $\mathcal{T}(\mathbf{x})$ to the client, and outputs nothing to the provider and cloud servers.
\end{enumerate}

\end{definition}

\begin{definition}
\label{def:security-def}
A protocol $\Pi$ securely realizes the $\mathcal{F}_{\mathsf{SecODT}}$ in the semi-honest adversary setting with static corruption if the following guarantees are satisfied:

\begin{enumerate}[-]

\item \textbf{Corrupted provider.} A corrupted and semi-honest provider learns nothing about the values in the client's feature vector $\mathbf{x}$. Formally, a probabilistic polynomial time (PPT) simulator $\mathsf{Sim}_{\mathcal{P}}$ should exist so that $\mathsf{View}^{\Pi}_{\mathcal{P}} \mathop  \approx \limits^c \mathsf{Sim}_{\mathcal{P}}(\mathcal{T})$, where $\mathcal{P}$ denotes the provider and $\mathsf{View}^{\Pi}_{\mathcal{P}}$ refers to the view of $\mathcal{P}$ in the real-world execution of the protocol $\Pi$.

\item \textbf{Corrupted cloud server.} A corrupted and semi-honest cloud server $\mathcal{C}_m$ ($m\in\{0,1\}$) learns no information about the client's feature vector $\mathbf{x}$ and the provider's decision tree $\mathcal{T}$. Formally, a PPT simulator $\mathsf{Sim}_{\mathcal{C}_i}$ should exist such that $\mathsf{View}^{\Pi}_{\mathcal{C}_m} \mathop  \approx \limits^c \mathsf{Sim}_{\mathcal{C}_m}$, where $\mathsf{View}^{\Pi}_{\mathcal{C}_m}$ denotes the view of the cloud server $\mathcal{C}_m$ in the real-world execution of the protocol $\Pi$.
Note that the two cloud servers have no input and output according to the $\mathcal{F}_{\mathsf{SecODT}}$. Since they are non-colluding, $\mathcal{C}_0$ and $\mathcal{C}_1$ cannot be corrupted by the adversary at the same time. 

\item \textbf{Corrupted client.} A corrupted and semi-honest client learns no information about the provider's decision tree other than generic meta-parameters as stated before. Formally, a PPT simulator $\mathsf{Sim}_{\mathcal{U}}$ should exist such that $\mathsf{View}^{\Pi}_{\mathcal{U}} \mathop  \approx \limits^c \mathsf{Sim}_{\mathcal{U}}(\mathbf{x},\mathcal{T}(\mathbf{x}))$, where $\mathcal{U}$ denotes the client, and $\mathsf{View}^{\Pi}_{\mathcal{U}}$ refers to the view of $\mathcal{U}$ in the real-world execution of the protocol $\Pi$.

\end{enumerate}
\end{definition}

\begin{theorem}
\label{thm:security-guarantee}
Our protocol is a secure realization of the ideal functionality $\mathcal{F}_{\mathsf{SecODT}}$ according to Definition \ref{def:security-def}. 
\end{theorem}

\begin{proof}
As per the security definition, we show the existence of a simulator for different corrupted parties (the provider, the client, and either of the cloud servers).

\begin{enumerate}[-]

\item \emph{Simulator for the corrupted provider}:
In the protocol $\Pi$, the provider only needs to supply the secret shares of the decision tree and receives no messages.
So, the simulator for the corrupted provider can thus be constructed in a dummy way by just outputting the input of the provider.
The output of $\mathsf{Sim}_{\mathcal{P}}(\mathcal{T})$ is identically distributed to the view $\mathsf{View}^{\Pi}_{\mathcal{P}}$ of the corrupted provider.

\item \emph{Simulator for a corrupted cloud server:} 
As two cloud servers have a symmetric role in our protocol $\Pi$, it suffices to show a simulator $\mathsf{Sim}_{\mathcal{C}_0}$ for $\mathcal{C}_0$.
It is noted that the input/output of $\mathcal{C}_0$ in our protocol are just secret shares of some data.
The security of additive secret sharing ensures that these secret shares are purely random and can be perfectly simulated by $\mathsf{Sim}_{\mathcal{C}_0}$ using random values.
For the interactions between the two cloud servers in different phases, they are in fact due to the calls of the oblivious array-entry read procedure (only used in the secure feature selection phase) and the secret-shared multiplication procedure based on Beaver's triples.
Let $\mathsf{Sim}^{\mathsf{ORead}}_{\mathcal{C}_0}$ and $\mathsf{Sim}^{\mathsf{SecMul}}_{\mathcal{C}_0}$ denote the corresponding simulators which can simulate a view indistinguishable from real view for $\mathcal{C}_0$ in the oblivious array-entry read procedure and the secret-shared multiplication procedure respectively.
It is noted that the existence of these two simulators has been proved in prior work.
With the existence of these simulators, $\mathsf{Sim}_{\mathcal{C}_0}$ first runs $\mathsf{Sim}^{\mathsf{ORead}}_{\mathcal{C}_0}$ with random strings as input in the secure feature selection phase. 
Then, $\mathsf{Sim}_{\mathcal{C}_0}$ sets the simulated output as the input to the subsequent phases.
On each call of the secret-shared multiplication procedure, $\mathsf{Sim}_{\mathcal{C}_0}$ runs $\mathsf{Sim}^{\mathsf{SecMul}}_{\mathcal{C}_0}$ in order.
Finally, $\mathsf{Sim}_{\mathcal{C}_0}$ combines and outputs in order the simulated view by $\mathsf{Sim}^{\mathsf{ORead}}_{\mathcal{C}_0}$ and $\mathsf{Sim}^{\mathsf{BMul}}$ on every secure multiplication as its output.
This generates the final simulator $\mathsf{Sim}_{\mathcal{C}_0}$ for the cloud server $\mathcal{C}_0$.

\item \emph{Simulator for the corrupted client:}
In the protocol $\Pi$, the client supplies secret shares of the feature vector $\mathbf{x}$ and only receives the two shares $[u^*]_0$ and $[u^*]_1$ of the inference result, from which the plaintext inference result $u^*$ is reconstructed as the output of the client.
The simulator thus only needs to simulate the messages (two shares) received by the client given his output $u^*$.
It can set a random value $r$ as one of the shares, say $[u^*]_1$, and $u^*-r$ for the other share $[u^*]_0$.
This is in fact just a direct application of additive secret sharing, the security of which ensures that $[u^*]_0$ and $[u^*]_1$ random values and indistinguishable from the shares received by the client.
The combination of the two simulated shares also produces $u^*$, which is the same as the output in the real protocol execution and thus guarantees correctness.
So the output of $\mathsf{Sim}_{\mathcal{U}}(\mathbf{x},\mathcal{T}(\mathbf{x}))$ is identically distributed to the view $\mathsf{View}^{\Pi}_{\mathcal{U}}$ of the corrupted client.
The proof of Theorem \ref{thm:security-guarantee} is completed.

\end{enumerate}

\end{proof}

\section{Experiments}
\label{sec:experiments}
\subsection{Setup}

Our protocol is implemented in C++.
For the oblivious transfer primitive, we rely on the libOTe library \cite{libOTe} which provides implementation of the protocol in \cite{KolesnikovKRT16}.
%
% We compile the code with CMake 3.10 and GNU g++ 7.5.0 compiler, with the optimization level set to O3.
Cloud-side experiments are conducted over two AWS t3.xlarge instances equipped with Intel Xeon Platinum 8175M CPU (2.50GHz and 16GB RAM): one in Europe (London) and one in US East (N. Virginia).
The average latency is $75.422$ ms and bandwidth is $161$ Mbits/s.
These two instances are situated in different regions for simulating the real-world scenario that the cloud servers are in different trust domains. 
The provider and the client are evaluated on an AWS t2.xlarge instance possessing an Intel Xeon E5-2676 v3 processor (2.40GHz and 16GB RAM).
We test with synthetic decision trees with realistic configurations, following prior works \cite{DWuFNL16,TaiMZC17,ZhengDWWN20}. 
The tree depth $d$ varies from $3$ to $17$, and the dimension $I$ of the feature vector varies from $9$ to $57$. 
We make comparison with the state-of-the-art prior work by Zheng et al. \cite{ZhengDWWN20} (the \emph{ZDWWN} protocol).

\subsection{Local-side Performance Evaluation}
%

\iffalse

\begin{figure}[t!]
\centerline{\includegraphics[width=0.45\textwidth]{./figures/client_commn}}
\caption{Communication performance of the client.}
\label{fig:client_commn_cost}
\end{figure}

\begin{figure}[t!]
\centerline{\includegraphics[width=0.45\textwidth]{./figures/client_compute}}
\caption{Computation performance of the client.}
\label{fig:client_compute_cost}
\end{figure}
\fi

\begin{figure}[t!]
\centerline{\includegraphics[width=0.46\textwidth]{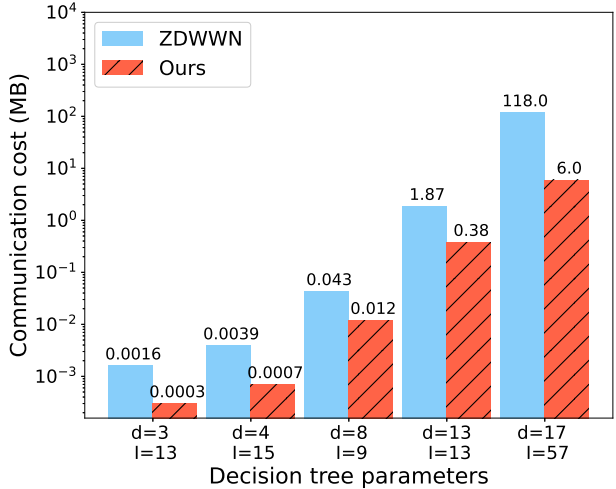}}
\caption{Communication performance of the provider.}
\label{fig:provider_commn_cost}
\end{figure}

\begin{figure}[t!]
\centerline{\includegraphics[width=0.46\textwidth]{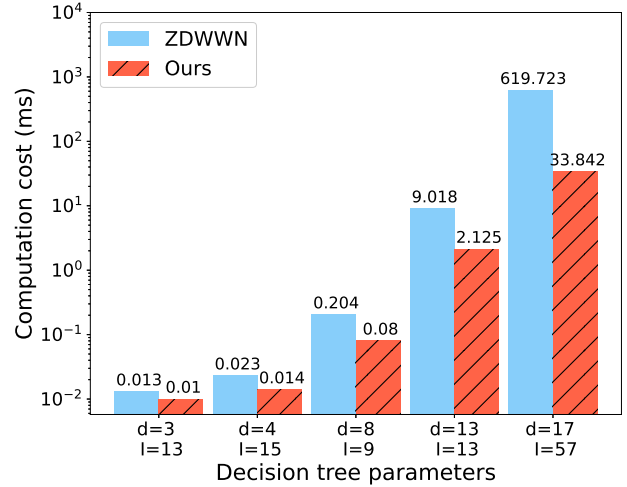}}
\caption{Computation performance of the provider.}
\label{fig:provider_compute_cost}
\end{figure}

We first evaluate the performance on the local side, i.e., the provider and the client.
Fig. \ref{fig:provider_commn_cost} and Fig. \ref{fig:provider_compute_cost} show the communication and computation costs of the provider for varying decisions trees, along with comparison with the ZDWWN protocol.
As the provider in our design just constructs an index vector of $O(J)$ size rather than a matrix of size $O(J\cdot I)$ as in the ZDWWN protocol, he can enjoy significant cost savings.
For different decision trees being tested, the communication cost of the provider ranges from $0.0003$ MB to $6$ MB in our system, while it is from $0.0016$ MB to $118$ MB in the ZDWWN protocol.
In terms of the provider's running time, it varies from $0.01$ ms to $33.842$ ms in our system, while it is from $0.013$ ms to $619.723$ ms in the ZDWWN protocol.
Overall, our system can offer the provider up to $19 \times$ savings (average of $7 \times$) in communication.
In computation, our system can offer the provider up to $18 \times$ savings (average of $5 \times$).
In Table \ref{table:client-cost}, we give the communication cost and computation cost of the client respectively.
It is noted that the client has minimal costs which only scales with the dimension of the feature vector.
Recall that the client in our system has the same cost as in the ZDWWN protocol, given that the polynomial-based mechanism is used in the phase of secure inference generation which allows the client to only receive the two shares of the inference result.

\begin{table}[!t]
\centering
\caption{Performace of the Client}
\begin{tabular}{@{}cccc@{}}
\toprule
$d$  & $I$  & Computation (ms) & Communication (KB) \\ \midrule
3  & 13 & 0.0057           & 0.22               \\ 
4  & 15 & 0.0060           & 0.25               \\ 
8  & 9  & 0.0054           & 0.16               \\ 
13 & 13 & 0.0056           & 0.22               \\ 
17 & 57 & 0.0098           & 0.91               \\ \bottomrule
\end{tabular}
\label{table:client-cost}
\end{table}

\subsection{Cloud-side Performance Evaluation}
We now examine the performance on the cloud side.
Firstly, we show in Fig. \ref{fig:cloud_commn_cost} the amount of data transferred between the cloud servers in our system and make comparison with the ZDWWN protocol.
Our design has relatively higher communication cost (average of $1.9\times$) than the ZDWWN protocol, which is mainly due to the sue of OT in the secure feature selection phase and the secret-shared multiplications in the secure inference generation phase.
We emphasize that such overhead in these two phases is the trade-off for the substantial efficiency improvement on the provider side, which is the first priority in our design philosophy.

\noindent\textbf{Significantly reduced overall online cloud service latency.} On another hand, it is noted that the overall end-to-end online inference latency in our system is still much less than the ZDWWN protocol, as shown in Fig. \ref{fig:cloud_compute_cost}.
This means that compared with the ZDWWN protocol, our system provides much better service experience for the client and fits much better into the practical realm due to the capability in giving faster response.
In particular, our new design is up to $8\times$ (average of $5 \times$) faster than the ZDWWN protocol.
Such efficiency is attributed to the significant optimization on the secure decision node evaluation phase, where the round complexity is largely reduced from linear (as in the ZDWWN protocol) to logarithmic.
A breakdown of the overall cloud-side online inference latency is given in Table \ref{table:cost}.

\begin{figure}[t!]
\centerline{\includegraphics[width=0.45\textwidth]{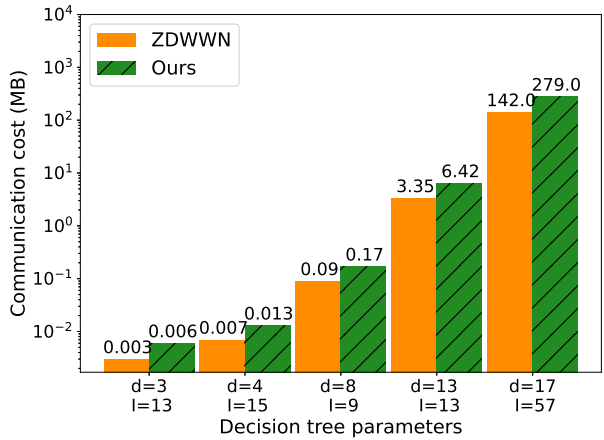}}
\caption{Communication performance at the cloud.}
\label{fig:cloud_commn_cost}
\end{figure}

\section{Related Work}
\label{sec:related_work}

There has been some work on secure decision tree inference \cite{BostPTG15,DWuFNL16,TaiMZC17,Cock17,TuenoKK19,TuenoBK20}.
Most of prior works \cite{BostPTG15,DWuFNL16,TaiMZC17,TuenoKK19,Cock17,TuenoBK20} focus on the non-outsourcing setting where a customized protocol is designed for running between the provider and the client.
For example, in \cite{TaiMZC17}, to achieve secure feature selection, the client sends to the provider the ciphertext of the feature vector under homomorphic encryption, and then the provider directly selects the ciphertext of each feature for each decision node based on his plaintext selection mapping.
Whether these protocols can be effectively adapted to the outsourcing setting remains largely unclear, since the outsourced service requires operations to conducted over encrypted decision tree and feature vector from the very beginning and also raises more design considerations for security and functionality.
Moreover, many protocols make use of heavy cryptographic tools (e.g., fully/partially homomorphic encryption, garbled circuits, and ORAM) in the latency-sensitive online interactions.
Although the protocol in \cite{Cock17} uses secret sharing, it is yet designed to fully and inefficiently work on binary representations of the decision tree as well as the feature vector provided from the very beginning, with all the secure processing conducted at bitwise level. 
So their protocol is also not directly adaptable for efficient secure outsourcing.

Very recently, the work \cite{ZhengDWWN20} presents the first design tailored for secure outsourcing of decision tree inference, which runs under the two-server model and only makes use of additive secret sharing to securely realize the various components for the online execution of the service.
As an initial attempt, however, its performance is yet to be optimized.
Our new highly efficient design presents significant optimizations which largely improves the overall online end-to-end latency of the secure inference service provided by the cloud, as well as the provider's performance.

\begin{figure}[t!]
\centerline{\includegraphics[width=0.45\textwidth]{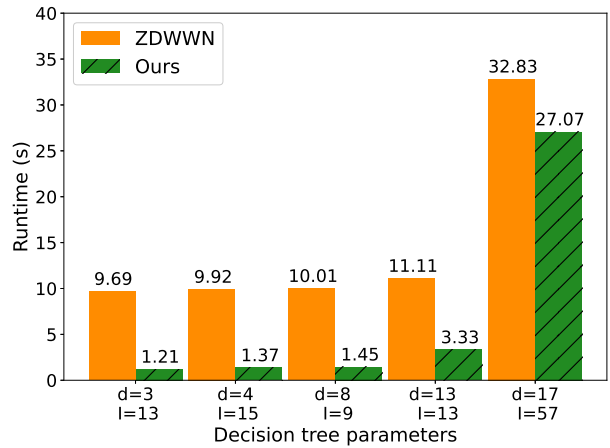}}
\caption{Overall \emph{end-to-end} online runtime performance at the cloud over realistic WAN.}
\label{fig:cloud_compute_cost}
\end{figure}

\begin{table*}[t!]
\centering
\caption{Breakdown of Runtimes in Different Phases (in seconds) at the Cloud Servers, over WAN}
\setlength{\tabcolsep}{0.7em}
\renewcommand{\arraystretch}{1.2}
\begin{tabular}{@{}cccccccccc@{}}
\toprule
\multicolumn{2}{c}{Parameters} & \multicolumn{2}{c}{Secure Feature Selection} & \multicolumn{2}{c}{Secure Decision Node Evaluation} & \multicolumn{2}{c}{Secure Inference Generation}  \\ \midrule
$d$                & $I$                & ZDWWN          & Ours        & ZDWWN         & Ours         & ZDWWN          & Ours        \\ \midrule
3                & 13               &     0.151          &    \textbf{0.527}         &   9.38           &  \textbf{0.529}            &   0.154             & \textbf{0.154}                        \\
4                & 15               &    0.156           &    \textbf{0.529}         &    9.454           &  \textbf{0.53}           &   0.306             &\textbf{0.306}                    \\
8                & 9                &     0.167          &    \textbf{0.533}         &    9.456           &  \textbf{0.532}            &  0.383              & \textbf{0.383}                       \\
13               & 13               &        0.393       &     \textbf{0.91}        &      10.03         &   \textbf{1.735}           &  0.69              & \textbf{0.69}                       \\
17               & 57               &     19.376          &      \textbf{21.006}       &   11.669            &   \textbf{4.281}           & 1.785
               &         \textbf{1.785}
              \\ \bottomrule
\end{tabular}
\label{table:cost}
\end{table*}

Our work is also related to the line of work (e.g., \cite{BostPTG15,LiuJLA17,JuvekarVC18,MohasselZ17,MLSZP20}, to just list a few) on securely evaluating other machine learning models, such as hyperplane decision \cite{BostPTG15}, Na{\"{\i}}ve Bayes \cite{BostPTG15}, neural networks \cite{LiuJLA17,JuvekarVC18,MohasselZ17,MLSZP20}.    
The common blueprint therein is to build specializaedprotocols tailored for the specific computation required by different models through different cryptographic techniques.
For example, the work of Liu et al. \cite{LiuJLA17} supports secure neural network evaluation using secret sharing and garbled circuits;
the work of Juvekar et al. \cite{JuvekarVC18} relies on highly customized use of homomorphic encryption and garbled circuits to support low latency in secure neural network evaluation. 
Most of these works operate under the non-outsourcing and aim to protect privacy for the model and client input.
There are some works \cite{MohasselZ17,NikolaenkoWIJBT13ridge} also operating under the two-server model as in this work, with the tailored support for secure evaluation of models such as linear regression, logistic regression, and neural networks.
Some recent efforts have also been presented on secure machine learning under the three-server \cite{MohasselR18,WaghGC19}/four-server \cite{Chaudhari2020} model (for models other than decision trees), where three/four servers have to engage in the online interactions.

\section{Conclusion}
\label{sec:conclusion}

In this paper, we design, implement, and evaluate a new system that allows highly efficient secure outsourcing of decision tree inference.
Through the synergy of several delicate optimizations which securely shift most workload of the provider to the cloud and reduce the communication round complexities between the cloud servers, our system significantly improves upon the state-of-the-art prior work.
Extensive experiments demonstrated that compared with the state-of-the-art, our new system achieves up to $8 \times$ better online end-to-end inference latency between the cloud servers over realistic WAN, as well as allows the provider to enjoy $19 \times$ savings in communication cost and $18 \times$ savings in computation cost.

\balance
\bibliographystyle{IEEEtran}
\bibliography{references}

\end{document}